\documentclass[letterpaper,11pt]{article}

\usepackage{amsmath}
\usepackage{amsfonts}
\usepackage{amssymb}
\usepackage{amsthm}
\usepackage[usenames]{color}
\usepackage{fullpage}
\usepackage{xspace}
\usepackage{url,ifthen}
\usepackage{srcltx}
\usepackage{multirow}
\usepackage{boxedminipage}
\usepackage[margin=1.2in]{geometry}
\usepackage{nicefrac}
\usepackage{xspace}
\usepackage{graphicx}
\usepackage{color}
\definecolor{DarkGreen}{rgb}{0.1,0.5,0.1}
\definecolor{DarkRed}{rgb}{0.5,0.1,0.1}
\definecolor{DarkBlue}{rgb}{0.1,0.1,0.5}

\usepackage[small]{caption}

\usepackage[pdftex]{hyperref}
\hypersetup{
    unicode=false,          % non-Latin characters in Acrobat¿s bookmarks
    pdftoolbar=true,        % show Acrobat toolbar?
    pdfmenubar=true,        % show Acrobat menu?
    pdffitwindow=false,      % page fit to window when opened
    pdfnewwindow=true,      % links in new window
    colorlinks=true,       % false: boxed links; true: colored links
    linkcolor=DarkRed,          % color of internal links
    citecolor=DarkGreen,        % color of links to bibliography
    filecolor=DarkRed,      % color of file links
    urlcolor=DarkBlue,          % color of external links
    pdftitle={},
    pdfauthor={},
}

\def\draft{0}

\def\submit{0}

\ifnum\draft=1 % show authors' note if draft
    \def\ShowAuthNotes{1}
\else
    \def\ShowAuthNotes{0}
\fi

\ifnum\submit=1
\newcommand{\forsubmit}[1]{#1}
\newcommand{\forreals}[1]{}
\else
\newcommand{\forreals}[1]{#1}
\newcommand{\forsubmit}[1]{}
\fi

\ifnum\ShowAuthNotes=1
\newcommand{\authnote}[2]{{ \footnotesize \bf{\color{DarkRed}[#1's Note:
{\color{DarkBlue}#2}]}}}
\else
\newcommand{\authnote}[2]{}
\fi

\newcommand{\Mnote}[1]{{\authnote{Moritz} {#1}}}

\newtheorem{theorem}{Theorem}[section]

\newtheorem{remark}[theorem]{Remark}
\newtheorem{lemma}[theorem]{Lemma}
\newtheorem{corollary}[theorem]{Corollary}
\newtheorem{proposition}[theorem]{Proposition}

\newtheorem{claim}[theorem]{Claim}
\newtheorem{fact}[theorem]{Fact}

\theoremstyle{definition}
\newtheorem{definition}[theorem]{Definition}

\newcommand{\chapterref}[1]{\hyperref[ch:#1]{Chapter~\ref{ch:#1}}}
\newcommand{\claimlabel}[1]{\label{claim:#1}}
\newcommand{\claimref}[1]{\hyperref[claim:#1]{Claim~\ref{claim:#1}}}
\newcommand{\corollarylabel}[1]{\label{cor:#1}}
\newcommand{\corollaryref}[1]{\hyperref[cor:#1]{Corollary~\ref{cor:#1}}}
\newcommand{\definitionlabel}[1]{\label{def:#1}}
\newcommand{\definitionref}[1]{\hyperref[def:#1]{Definition~\ref{def:#1}}}
\newcommand{\equationlabel}[1]{\label{eq:#1}}
\newcommand{\equationref}[1]{\hyperref[eq:#1]{Equation~\ref{eq:#1}}}
\newcommand{\factlabel}[1]{\label{fact:#1}}
\newcommand{\factref}[1]{\hyperref[fact:#1]{Fact~\ref{fact:#1}}}
\newcommand{\figurelabel}[1]{\label{fig:#1}}
\newcommand{\figureref}[1]{\hyperref[fig:#1]{Figure~\ref{fig:#1}}}
\newcommand{\itemlabel}[1]{\label{item:#1}}
\newcommand{\itemref}[1]{\hyperref[item:#1]{Item~\ref{item:#1}}}
\newcommand{\lemmalabel}[1]{\label{lem:#1}}
\newcommand{\lemmaref}[1]{\hyperref[lem:#1]{Lemma~\ref{lem:#1}}}

\newcommand{\propref}[1]{\hyperref[prop:#1]{Proposition~\ref{prop:#1}}}

\newcommand{\propositionref}[1]{\hyperref[prop:#1]{Proposition~\ref{prop:#1}}}

\newcommand{\remarkref}[1]{\hyperref[rem:#1]{Remark~\ref{rem:#1}}}
\newcommand{\sectionlabel}[1]{\label{sec:#1}}
\newcommand{\sectionref}[1]{\hyperref[sec:#1]{Section~\ref{sec:#1}}}
\newcommand{\theoremlabel}[1]{\label{thm:#1}}
\newcommand{\theoremref}[1]{\hyperref[thm:#1]{Theorem~\ref{thm:#1}}}

\usepackage[T1]{fontenc}
\usepackage{times}
\usepackage[varg]{txfonts} % varg - uses nicer g,v,y,w
\renewcommand{\mathbb}{\varmathbb}

\usepackage{microtype}

\newcommand{\Esymb}{\mathbb{E}}
\newcommand{\Psymb}{\mathbb{P}}

\DeclareMathOperator*{\E}{\Esymb}

\DeclareMathOperator*{\ProbOp}{\Psymb r}
\renewcommand{\Pr}{\ProbOp}

\newcommand{\nfrac}{\nicefrac}

\newcommand{\mper}{\,.}
\newcommand{\mcom}{\,,}

\renewcommand{\hat}{\widehat}

\newcommand{\cM}{{\cal M}}

\renewcommand{\leq}{\leqslant}
\renewcommand{\le}{\leqslant}
\renewcommand{\geq}{\geqslant}
\renewcommand{\ge}{\geqslant}
\newcommand{\Set}[1]{\left\{#1\right\}}
\newcommand{\bits}{\{0,1\}}

\newcommand{\R}{\mathbb{R}}

\usepackage{bm}

\newcommand{\polylog}{{\rm polylog}}

\renewcommand{\epsilon}{\varepsilon}

\newcommand{\remove}[1]{}

\newcommand{\range}{\mathrm{range}}
\newcommand{\PFP}{\textrm{\small PFP}\xspace}

\title{Beating Randomized Response on Incoherent Matrices}
\author{Moritz Hardt\thanks{IBM Research Almaden. Email: {\tt
mhardt@us.ibm.com}}
\and Aaron Roth\thanks{Department of Computer and Information Sciences, University of Pennsylvania. Email: {\tt aaroth@cis.upenn.edu}}}

\begin{document}

\maketitle

\begin{abstract}
Computing accurate low rank approximations of large matrices is a fundamental
data mining task. In many applications however the matrix contains sensitive
information about individuals. In such case we would like to release a low
rank approximation that satisfies a strong privacy guarantee such as
differential privacy. Unfortunately, to date the best known algorithm for this
task that satisfies differential privacy is based on naive \emph{input
perturbation} or \emph{randomized response}: Each entry of the matrix is
perturbed independently by a sufficiently large random noise variable, a low
rank approximation is then computed on the resulting matrix.

We give (the first) significant improvements in accuracy over randomized
response under the natural and necessary assumption that the matrix has
\emph{low coherence}.  Our algorithm is also very efficient and finds a
constant rank approximation of an $m\times n$ matrix in time $O(mn).$ Note
that even generating the noise matrix required for randomized response already
requires time $O(mn).$
\end{abstract}

\vfill
\thispagestyle{empty}
\pagebreak

\section{Introduction}
Consider a large $m\times n$ matrix $A$ in which rows correspond to
individuals, columns correspond to movies, and the non-zero entry in $A(i,j)$
represent the rating that individual $i$ has given to movie $j$.  Such a data
set shares two important characteristics with many other data sets:
\begin{enumerate}
\item It can be represented as a \emph{matrix} with very different dimensions. There are many more people than movies, so $n \gg m$
\item It is composed of \emph{sensitive information}: the rating that an individual gives to a particular movie (and the very fact that he watched said movie) can be possibly compromising information.
\end{enumerate}
Nevertheless, although we want to reveal little about the existence of
individual ratings in this data set, it might be extremely useful to
be able to allow data analysts to mine such a matrix for statistical
information. Even while protecting the privacy of individual entries, it might
still be possible to release another matrix that encodes a great deal of
information about the original data set. For example, we might hope to be able
to recover the cut structure of the corresponding rating graph, perform
principal component analysis (PCA), or apply some other data mining technique.

Indeed, this example is not merely theoretical. Data of exactly this form was
released by Netflix as part of their competition to design improved
recommender systems. Spectral methods such as PCA were commonly used on this
dataset, and privacy concerns were acknowledged: Netflix attempted to
``anonymize'' the dataset in an ad-hoc way. Following this supposedly
anonymized release, Naranyanan and Shmatikov \cite{NarayananS08} were able to
re-identify many individuals in the dataset by cross-referencing the reviews
with publicly available reviews in the internet movie database. As a result of
their work, a planned second Netflix challenge was canceled. The story need
not have ended this way however -- the formal privacy guarantee known as
\emph{differential privacy} could have prevented the attack of
\cite{NarayananS08}, and indeed, McSherry and Mironov \cite{McSherryM09}
demonstrated that many of the recommender systems proposed in the competition
could have been implemented in a differentially private way.
\cite{McSherryM09} make use of private low-rank matrix approximations using
input perturbation methods. In fact, it is not possible to generically improve
on input perturbation methods for all matrices without violating \emph{blatant
non privacy} \cite{DinurN03}. Nevertheless,  in this paper, we give the first
algorithms for low rank matrix approximation with performance guarantees that
are significantly better than input perturbation, under certain commonly
satisfied conditions \emph{which are already assumed} in prior work on
non-private low-rank matrix approximation.

In this paper, we consider the problem of privately releasing accurate
low-rank approximations to datasets that can be represented as matrices. Such
matrix approximations are one of the most fundamental building blocks for
statistical analysis and data mining, with key applications including latent
semantic indexing and principle component analysis. We provide theorems
bounding the accuracy of our approximations as compared to the optimal low
rank approximations in the Frobenius norm. The classical Eckart-Young theorem
asserts that the
optimal rank-$k$ approximation of a matrix $A$ (in either the Frobenius or
Spectral norms) is obtained by computing the singular value decomposition $A =
U\Sigma V^T$, and releasing the \emph{truncated} SVD $A_k = U\Sigma_kV^T$,
where in $\Sigma_k$, all but the top $k$ singular values have been zeroed out.
Computing the SVD of a matrix takes time $O(mn^2)$. In addition to offering
privacy guarantees, our algorithm is also extremely efficient: it requires
only elementary matrix operations and simple noisy perturbations, and for
constant $k$ takes time only $O(mn)$. This represents a happy confluence of
the two goals of privacy and efficiency. Normally, the two are at odds, and
differentially private algorithms tend to be (much) less efficient than their
non-private counterparts. In this case, however, we will see that some algorithms
for fast approximate low-rank matrix approximation are much more amenable
to a private implementation than their slower counterparts.

Computing low rank matrix approximations privately has been considered at
least since \cite{BlumDMN05}, and to date, no algorithm has improved over simple
input  perturbation, which achieves an error (when compared with the best rank
$k$ approximation $A_k$) in Frobenius norm of $\Theta(\sqrt{k(n + m)})$. Although this error is optimal without making any assumptions on the matrix, this
error can be prohibitive when the best rank $k$ approximation is actually very
good: when $\|A-A_k\|_F \ll \sqrt{k(n+m)}.$ That is, exactly in the case when a
low rank approximation to the matrix would be most useful.  We give an
algorithm which improves over input perturbation under the conditions that $m
\ll n$ and that the \emph{coherence} of the matrix is small: roughly, that no
single row of the matrix is too significantly correlated with any of the right singular
vectors of the matrix. Equivalently, no left singular vector has large
correlation with one of the standard basis vectors. Low coherence is a
commonly studied and satisfied condition. For example, Candes and Tao,
motivated by the same Netflix Prize dataset re-identified by
\cite{NarayananS08}, consider the problem of matrix completion under low
coherence conditions \cite{CandesT10}. They show that matrix completion is
possible under low coherence assumptions, and that several reasonable random
matrix models exhibit a strong incoherence property.  Notably, \cite{CandesT10}
were not concerned with privacy at all: they viewed low coherence as a natural
assumption satisfied for datasets resembling the Netflix prize data that could
be leveraged to obtain stronger utility guarantees. This represents a second
happy confluence of the goals of data privacy and utility: low coherence is an
assumption that others \emph{already} make free of privacy concerns
in order to improve the state of the art in data analysis.  We show that the
same assumption can simultaneously be leveraged for data privacy. In
retrospect, low coherence is also an extremely natural condition in the
context of privacy, although one that has not previously been considered in
the literature. If a matrix fails to have low coherence, then intuitively the
data of individual rows of the matrix is encoded closely in individual
singular vectors. If it does have low coherence, no small set of singular
vectors can be used to encode any row of the matrix with high accuracy, and
intuitively, low rank approximations reveal less local information about
particular entries of the matrix.

The problem we solve is the following: Given a matrix $A$ and a target rank
$k$ we privately compute and release a rank $O(k)$ matrix $B$ such that
$\|A-B\|_F$ is not much larger than $\|A-A_k\|_F$, where $A_k$ is the
\emph{optimal} rank $k$ approximation to $A$, and $\|\cdot\|_F$ is the
\emph{Frobenius norm}. The quality of the approximation depends on several
factors, including $n$, $m$, the desired rank $k$, and the coherence of the
matrix. Our approach improves over input perturbation when the matrix
coherence is small.

Our algorithm promises $(\epsilon,\delta)$-\emph{differential privacy}
\cite{DworkMNS06} with respect to changes of any single row of magnitude $1$ in
the $\ell_2$-norm.  This is only stronger than the standard notion of changing
any single entry in the matrix by a unit amount. In the very special case of
the matrix representing a (possibly unbalanced) graph, this captures (for
example) the addition or removal of a single edge. Therefore in this case our
algorithm is promising \emph{edge privacy} rather than \emph{vertex privacy}.
From a privacy point of view, this is less desirable than vertex privacy, but
is still a strong guarantee which is appropriate in many settings. We note
that edge privacy is well studied with respect to graph problems (see, e.g.
\cite{NissimRS07,GuptaLMRT10, GuptaRU11}), and we do not know of any
algorithms with non-trivial guarantees on graphs that promise vertex privacy,
nor any algorithms in the more general case of matrices that promise privacy
with respect to entire rows.

\subsection{Our results}

We start with our first algorithm that improves over randomized response on
matrices of small \emph{$C$-coherence}. We say that an $m\times n$
matrix $A$ has \emph{coherence $C,$} if no row has Euclidean norm more than
$C\cdot \|A\|_F/\sqrt{m},$ i.e., more than $C$ times the the typical row
norm.This parameter varies between $1$ and $\sqrt{m},$ since no row can have
Euclidean norm more than $\|A\|_F.$ Intuitively the condition says that no
single row contributes too significantly to the Frobenius norm of the matrix.

\begin{theorem}[Informal version of \theoremref{C-approx}]
\theoremlabel{informal1}
There is an $(\epsilon,\delta)$-differentially private algorithm which given a
matrix $A\in\R^{m\times n}$ of coherence $C$ such that $n\ge m$
computes a rank $2k$ matrix $B$ such that with probability $9/10,$
\[
\|A-B\|_F
\le O\left(\|A-A_k\|_F\right)
+ O_{\epsilon,\delta}\left(
\sqrt{km} + \sqrt{kn}\cdot \frac{\sqrt{Ck\|A\|_F}}{(nm)^{1/4}}\right)\mper
\]
Moreover, the algorithm runs in time $O(kmn).$
\end{theorem}
Hidden in the $O_{\epsilon,\delta}$-notation is a factor of
$O(\log(k/\delta)/\epsilon)$ that depends on the privacy parameters.
Usually, $\delta\ll 1/k$ so that
$\log(k/\delta)\le 2\log(1/\delta).$
To understand the error bound note that the first term is proportional to the
best possible approximation error $\|A-A_k\|_F$ of any rank $k$
approximation. In particular, this term is optimal up to constant factors.
The second term expresses a more interesting phenomenon. Recall that we assume
$n\gg m$ so that $\sqrt{kn}$ would usually dominate $\sqrt{km}$ except that the
the $\sqrt{kn}$ term is multiplied by a factor which can be very
small if the matrix has low coherence and is not too dense. For example, when
$k=O(1),$ $C=O(1)$ and $\|A\|_F=O(\sqrt{n}),$ the error is roughly
$O(\sqrt{m}+\sqrt{n}/m^{1/4})$ which can be as small as $O(n^{3/8})$ depending
on the magnitude of $m.$ However, already in a much wider range of parameters we
observe an error of $o(\sqrt{kn}).$ In fact, in \sectionref{netflix} we
illustrate why the Netflix data satisfies the assumptions made
here and why they are likely to hold in other recommender systems.

When $\|A\|_F\ge \sqrt{n},$ the previous theorem cannot improve on randomized
response by more than a factor of $O(m^{1/4}).$ Our next theorem uses a
stronger but standard notion of coherence known as $\mu_0$-coherence. We
defer a formal definition of $\mu_0$-coherence to \sectionref{incoherent}, but
we remark that this parameter varies between $1$ and $m.$ Using this notion we
are able to obtain improvements roughly of order $O(\sqrt{m}).$

\begin{theorem}[Informal version of \theoremref{mu-approx}]
\theoremlabel{informal2}
There is an $(\epsilon,\delta)$-differentially private algorithm which given a
matrix $A\in\R^{m\times n}$ with $n\ge m$ and of $\mu_0$-coherence $\mu$ and
rank $r\ge 2k$ computes a rank $2k$ matrix $B$ such that with
probability $9/10,$
\[
\|A-B\|_F
\le O\left(\|A-A_k\|_F\right)
+ O_{\epsilon,\delta}\left(
\sqrt{km} + \sqrt{kn}\cdot \sqrt{\frac{\mu kr}{m}}\right)\mper
\]
Moreover, the algorithm runs in time $O(kmn).$
\end{theorem}

The hidden factor here is the same as before.   Note that
when $\mu kr=\polylog(n),$ the theorem can lead to an error bound $\tilde
O\left(n^{1/4}\right)$ depending on the magnitude of $m.$ Note that this is
roughly the square root of what randomized response would give. But again under
much milder assumptions on the coherence, the error remains
$o\left(\sqrt{kn}\right).$ Notably, Candes and Tao~\cite{CandesT10} work
with a stronger incoherence assumption than what is needed here.
Nevertheless they show that even their stronger assumption is
satisfied in a number of reasonable random matrix models. A slight
disadvantage of the error bound in \theoremref{informal2} is that the actual
rank~$r$ of the matrix enters the picture. \theoremref{informal2} hence cannot
improve over \theoremref{informal1} when the matrix has very large rank. We do
not know if the dependence on $r$ in the above bound is inherent or rather an
artifact of our analysis.

Finally, we remark that while our result depends on the $\mu_0$-coherence of the input
matrix, our algorithm does not require knowledge or estimation of the
$\mu_0$-coherence of the input matrix. The only parameters provided to the algorithm
are the target rank and the privacy parameters.

\paragraph{Reconstruction attacks and tightness of our results.}
As it turns out, existing work on ``blatant non-privacy'' and
reconstruction attacks~\cite{DinurN03} demonstrates that our results are
essentially tight under the given assumptions. To draw this connection, let us
first observe why \emph{input perturbation} cannot be improved without any
assumption on the matrix. To be more precise, by input perturbation we refer
to the method which simply perturbs each entry of the matrix with independent
Gaussian noise of magnitude
$O\left(\epsilon^{-1}\sqrt{\log(1/\delta)}\right)$, which is sufficient to
achieve $(\epsilon,\delta)$-differential privacy with respect to unit $\ell_2$
perturbations of the entire matrix. To obtain a rank $k$ approximation to the
original matrix, one can then simply compute the exactly optimal rank $k$
approximation to the perturbed matrix using the singular value decomposition,
which as one can show introduces error $O_{\epsilon,\delta}\left(\sqrt{km}+\sqrt{kn}\right)$
compared to the optimal rank $k$ approximation to the original matrix in the
Frobenius norm. First, let us observe that it is not possible in general to
have an algorithm which guarantees error in the Frobenius norm of
$o(\sqrt{kn})$ for \emph{every} matrix $A$, without violating \emph{blatant
non-privacy}\footnote{An algorithm $\cM$ is blatantly non-private if for every
database $D\in\bits^{n'}$ it is possible to reconstruct a $1-o(1)$ fraction of
the entries of $D$ exactly, given only the output of the mechanism $\cM(D)$.},
as defined by \cite{DinurN03}.  This is because there is a simple
reduction which starts with an $(\epsilon,\delta)$-differentially private
algorithm for computing rank $k$ approximations to matrices $A \in
\mathbb{R}^{m\times n}$ and gives an $(\epsilon,\delta)$-differentially
private algorithm which can be used to reconstruct almost every entry in any
database $D \in \{0,1\}^{n'}$ for $n' = k\cdot n$. It is known that
$(\epsilon,\delta)$-private mechanisms do not admit such reconstruction
attacks, 
%for all databases $D \in \{0,1\}^{n'}$
 and so the result is a lower
bound. The reduction follows from the fact that we can always encode a
bit-valued database $D \in \{0,1\}^{n'}$ for $n' = k\cdot n$ as $k$ rows of an
$m\times n$ matrix for any $m \geq k$, simply by zeroing out all additional
$m-k$ rows. Note that the resulting matrix only has rank $k$, and so the
optimal rank $k$ approximation to this matrix has \emph{zero}
error. If we could recover a  matrix $A'$ such that $\|A-A'\|_F =
o(\sqrt{kn})$, this would mean that for a typical nonzero row $A_i$ of the
matrix with $i \in [k]$, we would have $\|A_i - A'_i\|_2 = o(\sqrt{n})$, and $\|A_i -
A'_i\|_1 \leq o(n)$. Then, by simply rounding the entries, we could
reconstruct the original database $D$ in almost all of its entries, giving
blatant non-privacy as defined by \cite{DinurN03}.

What is happening in the above example? Intuitively, the problem is that in
the rank $k$ matrix we construct from $D$, the $k$ nonzero rows of the matrix
are encoded accurately by only $k$~right singular vectors. On the other hand, low
coherence implies that any $k$ right singular vectors poorly represent a set
of only $k$ rows. Hence, there is hope to circumvent the above impediment using a
low coherence assumption on the matrix. Indeed, this is precisely what
\theoremref{informal1} and \theoremref{informal2} demonstrate.
%On the other hand, the above
%example does not mean that we cannot obtain a rank $k$ Frobenius norm
%approximation of $o(\sqrt{kn})$ for any matrix! For example, if we have a
%matrix with $\gg k$ nonzero rows (even if it still has rank $O(k)$), it is not
%clear how to reconstruct the original database given only a rank $k$
%approximation to the matrix. More generally, it seems that if the singular
%vectors of the matrix have little correlation with the actual rows of the
%matrix, then computing good low-rank approximations need not correspond to
%reconstruction attacks. Indeed,
%In this paper, we identify \emph{low
%coherence} conditions already assumed outside of the privacy literature, which
%do in fact loosely correspond to low correlation between singular vectors and
%matrix rows. Under these conditions, we show that it is in fact possible to
%obtain error $o(\sqrt{kn})$ while preserving differential privacy. Our
%theorems are tight, in the sense that improving our dependence on the
%coherence condition would result in bounds of $o(\sqrt{n})$ for rank $k =
%O(1)$ matrices, which would violate blatant non-privacy.
%
Nevertheless, reconstruction attacks still lead to lower bounds even under low
coherence assumptions. Indeed, using the above ideas, the next proposition shows
that \theoremref{informal1} is essentially tight up to a factor
of~$O(\sqrt{k}).$ Since in many applications $k=O(1),$ this discrepancy
between our upper bound and the lower bound is often insignificant.
%proposition shows that in such cases
%\theoremref{informal1} is tight given the assumptions of the theorem.
%
\begin{proposition}
Any algorithm which given an $m\times n$ matrix $A$ of coherence $C$ outputs a
rank $k$ matrix $B$ such that with high probability
\[
\|A-B\|_F\le o\left(\sqrt{kn}\cdot\frac{\sqrt{C\|A\|_F}}{(nm)^{1/4}}\right)
\]
cannot satisfy $(\epsilon,\delta)$-differential privacy for
sufficiently small constants $\epsilon,\delta.$
\end{proposition}
\begin{proof}[Informal proof]
For the sake of contradiction, suppose there exists such an algorithm $\cM$
that satisfies $(\epsilon,\delta)$-differential privacy. Then consider a
randomized algorithm $\cM'\colon\bits^{n'}\to\R^{n'}$ which takes a data set
$D\in\bits^{n'}$ containing a sensitive bit for $n'=kn$ individuals and
encodes it as the $m\times n$ matrix $A_D$ which contains $D$ in its first $k$
rows and is $0$ everywhere else. $\cM'(D)$ then computes $\cM(A_D)$ and
outputs the projection of $\cM(A_D)$ onto the first $k$ rows (thought of as a
vector of length $n'=kn$).

We claim that $\cM'$ is
$(\epsilon,\delta)$-differentially privacy. This is because the map from $D$
to $A_D$ is sensitivity preserving and the post-processing computed on $M(A_D)$
preserves $(\epsilon,\delta)$-differential privacy of $M.$

On the other hand, we claim that $\cM'$ is blatantly non-private. To see this
note that the matrix $A_D$ has coherence $C=\sqrt{m/k}$ and $\|A\|_F\le\sqrt{kn}$
so that one can check that $\|A-\cM(A)\|_F\le o(\sqrt{kn})$ with high probability.
This implies that $\|D-\cM'(D)\|_2\le o(\sqrt{n'})$ with high probability. We therefore also have $\|D - \cM'(D)\|_1 \le o(n')$. But in
this case we can compute a data set $D'$ from the output of $M'(D)$ such that
$\|D-D'\|_0=o(n')$ by rounding. This is the definition of a reconstruction attack showing
that $\cM'$ is blatantly non-private. Since $(\epsilon,\delta)$-differential
privacy is known to prevent blatant non-privacy\footnote{See, e.g., the proof
of Theorem 4.1 in~\cite{De11}.} for sufficiently small
$\epsilon,\delta>0,$ this presents the contradiction we sought.
\end{proof}
A similar proof shows that error $o(\sqrt{n}\cdot \sqrt{\mu/m})$ (where $\mu$
is the $\mu_0$-coherence of the matrix) cannot be achieved with
$(\epsilon,\delta)$-differential privacy. This shows that also
\theoremref{informal2} is tight up to the exact dependence on~$k$ and~$r.$ We
leave it as an intriguing open problem to determine the exact interplay
between coherence and the other parameters.

%Note that this theorem is tight in the sense that if $k = 1$ and the matrix has just a single bit-valued row (i.e. $\|A\|_F \leq \sqrt{n}$), then improving our dependence on $C$ would give an error bound of $o(\sqrt{n})$ (since for every matrix, $C \leq \sqrt{m}$). But as we discussed earlier, this would violate blatant non-privacy.

%We note that just as with the last theorem, this theorem has a tight dependence on $\mu$: if the dependence on $\mu$ were improved, then reconstruction attacks could be mounted on single-row matrices.

\subsection{Techniques and proof overview}
Our algorithm is based on a random-projection algorithm of Halko, Martinsson
and Tropp~\cite{HalkoMT11},
which involves two steps: \emph{range finding} and \emph{projection}. The
range finding algorithm first computes $k$ Gaussian measurements of $A,$ which
we denote by $Y=A\Omega.$ Here, $A$ is $m\times n$ and $\Omega$ is $n\times
k.$ These measurements can be thought of as a random projection of the matrix
into a lower dimensional representation, i.e., $Y$ is $m\times k.$ The crux of
the analysis in~\cite{HalkoMT11} is in arguing that $Y$ already captures most
of the range of~$A.$ Hence, all that remains to be done is to compute the
orthonormal projection operator~$P_Y$ into the span of $Y,$ and to compute the
projection $P_YA.$ Note that $P_YA$ is now a $k$-dimensional approximation of
$A$ and since $Y$ closely approximated the range of $A,$ it must be a good
approximation, say, in the Frobenius norm.

The motivation of~\cite{HalkoMT11} was to obtain a fast low rank approximation
algorithm. Indeed, \cite{HalkoMT11} give a detailed theoretical analysis and
empirical evaluation of the algorithm's performance.

\paragraph{Step 1: Privacy preserving range finder and projection.}
We will leverage the algorithm of~\cite{HalkoMT11} to obtain improved accuracy
bounds in the privacy setting. As a first step, we need to be able to carry out
the range finding and projection step in a privacy preserving manner.
Our analysis proceeds by observing that the projection of
$A$ to $Y$ approximately preserves all of the $\ell_2$ row-norms of $A$, and
so we can apply a Gaussian perturbation to $Y$, rather than to $A$. (An
$m\times k$ standard Gaussian matrix has Frobenius norm $O(\sqrt{km})$,
which is now independent of $n$). The formal presentation of this part of the
argument appears in \sectionref{range-projection}. This step provides an
approximation to the range of $A$ which might already be useful for some applications,
but has not yet achieved our goal of computing a low rank approximation to $A$
itself. For this, we need the projection step discussed next.

\paragraph{Step 2: Controlling the projection matrix using low coherence.}
We then show that under our low-coherence
assumption on $A$, the entries of the projection matrix into the range of $Y$,
$P_Y$, must be small in magnitude.
Finally, when $P_Y$ has small entries, the final
projection step, of computing $P_YA$ has low sensitivity, and although we must
now again add a Gaussian perturbation of dimension $m \times n$, the magnitude
of the perturbation in each entry can be smaller than would have been
necessary under naive input perturbation.

In order to obtain bounds on the $\ell_\infty$-norm of the projection operator
we make crucial use of the low-coherence assumption. Here we describe the
proof strategy that leads to \theoremref{informal2}. \theoremref{informal1} is
somewhat easier to show and follows along similar lines.
The first observation is that the Gaussian measurements
taken by the range finding algorithm are mostly linear combinations of the top
left singular vectors of the matrix. But when the matrix $A$ has low coherence,
then its top left singular vectors must have very small correlation with the
standard basis. This means that the top singular vectors must have small
coordinates. As a result each of the Gaussian measurements we take must have
small $\ell_\infty$-norm relative to the magnitude of the measurement. Some
complications arise as we must add noise to the matrix $Y$ for privacy reasons and
then orthonormalize it using the Gram-Schmidt orthonormalization algorithm.
A key observation is that the noise matrix is generated independently of $Y$.
As a result, it must be the case that all columns of the noise matrix have
very small inner product with the columns of $Y.$ A careful technical argument
uses this observation in order to show that the effect of noise can be
controlled throughout the Gram-Schmidt orthonormalization. The result is a
projection matrix in which the magnitude of each entry is small whenever the
coherence of~$A$ was small to begin with.

The exact proof strategy depends on the notion of coherence that we work with.
Both notions we consider in this paper are presented and analyzed in
\sectionref{incoherent}. We then also show that small $\mu_0$-coherence is indeed a
stronger assumption than small $C$-coherence.

\subsection{Related Work}
\subsubsection{Differential Privacy}
We use as our privacy solution concept the by now standard notion of
\emph{differential privacy}, developed in a series of papers \cite{BlumDMN05,
ChawlaDMSW05,DworkMNS06}, and first defined by Dwork, McSherry, Nissim, and Smith
\cite{DworkMNS06}. The problem of privately computing low-rank approximations to
matrix valued data was one of the first problems studied in the differential
privacy literature, first considered by Blum, Dwork, McSherry, and Nissim
\cite{BlumDMN05}, who give an input perturbation based algorithm for computing
the singular value decomposition by directly computing the eigenvector
decomposition of a perturbed covariance matrix. Computing low rank
approximations is an extremely useful primitive for differentially private
algorithms, and indeed, McSherry and Mironov \cite{McSherryM09} used the algorithm
given in \cite{BlumDMN05} in order to implement and evaluate differentially
private versions of recommendation algorithms from the Netflix prize
competition.

Finding differentially private low-rank approximation algorithms with superior
theoretical performance guarantees to input perturbation methods has remained
an open problem. Beating input perturbation methods for arbitrary symmetric
matrices was recently explicitly proposed as an open problem in
\cite{GuptaRU11}, who showed that such algorithms would lead to the first
efficient algorithm for privately releasing synthetic data useful for
\emph{graph cuts} which improves over simple randomized response. Our work
does not resolve this open question because our results only improve over
input perturbation methods for matrices with unbalanced dimensions which
satisfy a low-coherence assumption, but is the first algorithm to improve over
\cite{BlumDMN05} under any condition.

\paragraph{Comparison to recent results of Kapralov, McSherry and Talwar.}

In a recent independent and simultaneous work, Kapralov, McSherry, and
Talwar~\cite{KapralovMT11} give a new polynomial-time algorithm for computing
privacy-preserving rank 1 approximations to symmetric, positive-semidefinite matrices. Their algorithm achieves
$(\epsilon,0)$-differential privacy under unit spectral norm perturbations to
the matrix. Their algorithm outputs a vector $v$ such that for all $\alpha > 0$, $\E[v^TAv] \geq (1-\alpha)\|A\| - O(n\log(1/\alpha)/(\epsilon\alpha))$ (where $\|\cdot\|$ denotes the spectral norm) and they show that this is nearly tight for $(\epsilon,0)$-differential privacy guarantees. Our results are therefore
strictly incomparable. In this work, the goal is to achieve error
$o(\sqrt{kn})$ (i.e. $o(\sqrt{n})$ for rank-1 approximations) assuming low coherence, (a stronger error bound) under
$(\epsilon,\delta)$-differential privacy (a weaker privacy guarantee) and without making any assumptions about symmetry or positive-semidefiniteness.

\subsubsection{Fast Computation of Low Rank Matrix Approximations}
There is also an extensive literature on randomized algorithms for computing
approximately optimal low rank matrix approximations, motivated by improving
the running time of the exact singular value decompositions. This literature
originated with the work of Papadimitriou et al \cite{PapadimitriouTRV98} and Frieze,
Kannan, and Vempala \cite{FriezeKV04}, who gave algorithms based on random
projections and column sampling (in both cases with the goal of decreasing the
dimension of the matrix). Achlioptas and McSherry \cite{AchlioptasM01} give fast
algorithms for computing low rank approximations based on randomly perturbing
the original matrix (which can be done to induce sparsity). Although
\cite{AchlioptasM01} pre-dated the privacy literature, some of the algorithms presented
in it can be viewed as privacy preserving, because perturbing the actual
matrix with appropriately scaled Gaussian noise is a privacy preserving
operation sometimes referred to as \emph{randomized response}. When
appropriately scaled (for privacy) Gaussian noise is added to an $m\times n$
matrix, it results in an algorithm for approximating the best rank $k$
approximation up to an additive error of $O(\sqrt{k(m+n)})$ in the Frobenius
norm.
%Prior to our work, this was the best known worst-case bound for
%privately computing low rank matrix approximations.

Our algorithms are most closely related to  the very recent work of Halko,
Martinsson, and Tropp \cite{HalkoMT11}, who give fast algorithms for computing
low rank approximations based on two steps: range finding, and projection. As
already discussed, in the first step, these algorithms project the matrix $A$
into an $m\times k$ matrix $Y$ which approximately captures the \emph{range}
of $A$. Then $A$ is projected into the range of $Y$, which yields a rank $k$
matrix which gives a good approximation to $A$ if a good rank-$k$
approximation exists. We will further discuss the algorithm of
\cite{HalkoMT11} and our modifications in the course of the paper.

%The first step of this algorithm is amenable to privacy analysis, because random projections preserve $(\ell_2)$ row norms, and so we can perturb the projected matrix instead of the original matrix, and avoid any dependence on the larger dimension $n$. We are able to make the second step privacy preserving as well without incurring error $\sqrt{n+m}$ when the entries of the projection matrix $P_Y$ are small. It is here that we take advantage of low-coherence assumptions, to argue that either the entries of $P_Y$ must be small, or that any large entries can be truncated with little loss in the final approximation quality.

\subsubsection{Low Coherence Conditions}

Low coherence conditions have been recently studied in a number of papers for
a number of matrix problems, and is a commonly satisfied condition on
matrices. Recently, Candes and Recht \cite{CandesR09} and Candes and Tao
\cite{CandesT10} considered the problem of \emph{matrix completion}. Matrix
completion is the problem of recovering all entries of a matrix from which
only a subset of the entries which have been randomly sampled. This problem is
inspired by the Netflix prize recommendation problem, in which a matrix is
given, with individuals on the rows, movies on the columns, and in which the
matrix entries correspond to individual movie ratings. The matrix provides
only a small number of movie ratings per individual, and the challenge is to
predict the missing entries in the matrix.  Clearly accurate matrix completion
is impossible for arbitrary matrices, but \cite{CandesR09,CandesT10} show the
remarkable result that it is possible under low coherence assumptions. Candes
and Tao \cite{CandesT10} also show that almost every matrix satisfies a low
coherence condition, in the sense that randomly generated matrices will be low
coherence with extremely high probability.

Talwalkar and Rostamizadeh recently used low-coherence assumptions for the
problem of (non-private) low-rank matrix approximation \cite{TalwalkarR10}. A common
heuristic for speeding the computation of low-rank matrix approximations is to
compute on only a small randomly chosen subset of the columns, rather than on
the entire matrix. \cite{TalwalkarR10} showed that under low-coherence assumptions
similar to those of \cite{CandesR09,CandesT10}, the spectrum of a matrix is in fact
well approximated by a small number of randomly sampled columns, and give
formal guarantees on the approximation quality of the sampling based
Nystr\"{o}m method of low-rank matrix approximation.

\subsection*{Acknowledgments}

We thank Boaz Barak, Anupam Gupta, and Jon Ullman for many helpful
discussions. We thank Sham Kakade for pointing us to the paper of
\cite{HalkoMT11}. We thank Michael Kapralov, Frank McSherry and Kunal Talwar
for communicating their result to us, and for useful
conversations about low rank matrix approximations. We wish to thank
Microsoft Research New England where part of this research was done.

\section{Preliminaries}
We view our dataset as a real valued \emph{matrix} $A\in\mathbb{R}^{m\times
n}.$ We sometimes denote the  $i$-th of a matrix by $A_{(i)}.$
Let
\begin{equation}
\mathcal{N} = \{P \in \mathbb{R}^{m\times n} : \textrm{there exists
an index } i \in [m] \textrm{ such that } \|P_{(i)}\|_2 \leq 1 \textrm{ and }
\|P_{(j)}\|_2 = 0 \textrm{ for all } j\neq i\}
\end{equation}
denote the set of matrices
that take $0$ at all values, except possibly in a single row, which has
Euclidean norm at most $1$.
\begin{definition} We say that two matrices $A,
A' \in \mathbb{R}^{m\times n}$ are \emph{neighboring} if $(A - A') \in
\mathcal{N}$.
\end{definition}
We use the by now standard privacy solution concept of differential privacy:
\begin{definition}
An algorithm $M\colon\mathbb{R}^{m\times n}\rightarrow R$ (where $R$ is some
arbitrary abstract range) is \emph{$(\epsilon,\delta)$-differentially private}
if for all pairs of neighboring databases $A, A' \in \mathbb{R}^{m\times n}$, and for
all subsets of the range $S \subseteq R$:
\[
\Pr\Set{M(A) \in S} \leq
\exp(\epsilon)\Pr\Set{M(A') \in S} + \delta\]
\end{definition}

We make use of the following useful facts about differential privacy.
\begin{fact}
If $M:\mathbb{R}^{m\times n}\rightarrow R$ is $(\epsilon,\delta)$-differentially private, and $M':R\rightarrow R'$ is an arbitrary randomized algorithm mapping $R$ to $R'$, then $M'(M(\cdot)):\mathbb{R}^{m\times n}\rightarrow R'$ is $(\epsilon,\delta)$-differentially private.
\end{fact}

The following useful theorem of Dwork, Rothblum, and Vadhan tells us how differential privacy guarantees compose.
\begin{theorem}[Composition \cite{DworkRV10}]
\label{thm:composition}
Let $\epsilon,\delta\in(0,1),\delta'>0.$
If $M_1, \ldots, M_k$ are each $(\epsilon,\delta)$-differentially private
algorithms, then the algorithm $M(A)
\equiv (M_1(A),\ldots,M_k(A))$ releasing the concatenation of the results of
each algorithm is $(k\epsilon, k\delta)$-differentially private. It is also $(\epsilon', k\delta + \delta')$-differentially private for:
$$\epsilon' < \sqrt{2k\ln(1/\delta')}\epsilon + 2k\epsilon^2$$
\end{theorem}

We denote the $1$-dimensional Gaussian distribution of mean $\mu$ and
variance $\sigma^2$ by $N(\mu,\sigma^2).$ We use $N(\mu,\sigma^2)^d$ to denote
the distribution over $d$-dimensional vectors with i.i.d. coordinates sampled from
$N(\mu,\sigma^2).$ We write $X\sim D$ to indicate that a
variable $X$ is distributed according to a distribution~$D.$ We note
the following useful fact about the Gaussian distribution.
\begin{fact}\factlabel{gaussian-sum}
If $g_i\sim N(\mu_i,\sigma_i^2),$ then
$\sum g_i \sim N\left(\sum_i\mu_i,\sum_i\sigma_i^2\right)\mper$
\end{fact}
The following theorem is well known folklore. We include a proof in the appendix for completeness.

\begin{theorem}[Gaussian Mechanism]
\label{thm:Gaussian}
Let $x, y \in \mathbb{R}^d$ be any two vectors such that $\|x-y\|_2 \leq c$.
Let $Y \in \mathbb{R}^d$ be an independent random draw from $N(0,\rho^2)^d$,
where $\rho = c\epsilon^{-1}\sqrt{\log(1.25/\delta)}$. Then for any $S
\subseteq \mathbb{R}^d:$ $$\Pr\Set{x + Y \in S} \leq \exp(\epsilon)\Pr[y + Y
\in S] + \delta$$
\end{theorem}

\paragraph{Vector and matrix norms.} We denote by $\|\cdot\|_p$ the
$\ell_p$-norm of a vector and sometimes use $\|\cdot\|$ as a shorthand for the
Euclidean norm. Given a real $m\times n$ matrix $A,$ we will work with the
\emph{spectral norm} $\|A\|_2$ and the Frobenius norm~$\|A\|_F$ defined as
\begin{equation}\textstyle
\|A\|_2  =\max_{\|x\|=1}\|Ax\|\qquad\text{and}
\qquad \|A\|_F =\sqrt{\sum_{i,j}a_{ij}^2}\mper
\end{equation}
For any $m\times n$ matrix $A$ of rank $r$ we have
$\|A\|_2\le\|A\|_F\le\sqrt{r}\cdot\|A\|_2\mper$
%\begin{fact}\label{fac:relations}
%\end{fact}
For a matrix $Y$ we denote by~$P_Y$
the orthonormal projection operator onto the range of $Y.$
\begin{fact}\label{fac:projection}
$P_Y=Y(Y^*Y)Y^{-1}$
\end{fact}

\begin{fact}[Submultiplicativity]\factlabel{sub}
For any $m\times n$ matrix $A$ and $n\times r$ matrix $B$ we have
%\begin{enumerate}
%\item
$\|AB\|_F\le\|A\|_F\cdot\|B\|_F.$
%\item $\|AB\|\le\|A\|\cdot\|B\|.$
%\end{enumerate}
\end{fact}

\begin{theorem}[Weyl]\label{thm:weyl}
For any $m\times n$ matrices $A,E,$ we have
$\left|\sigma_i(A+E)-\sigma_i(A)\right|\le\|E\|_2\mcom$ where $\sigma_i(M)$
denotes the $i$-th singular value of a matrix $M.$ where $\sigma_i(M)$ denotes
the $i$-th singular value of a matrix~$M.$
\end{theorem}

\section{Low-rank approximation via Gaussian measurements}

We will begin by presenting an algorithm of Halko, Martinsson and
Tropp~\cite{HalkoMT11} as described in \figureref{HMT}.  The algorithm
produces a rank $r+p$ approximation that already for $p\ge 2$ closely matches
the best rank~$r$ approximaton of the matrix in Frobenius norm. The guarantees
of the algorithm are detailed in \theoremref{HMT}.

\begin{figure}[h]
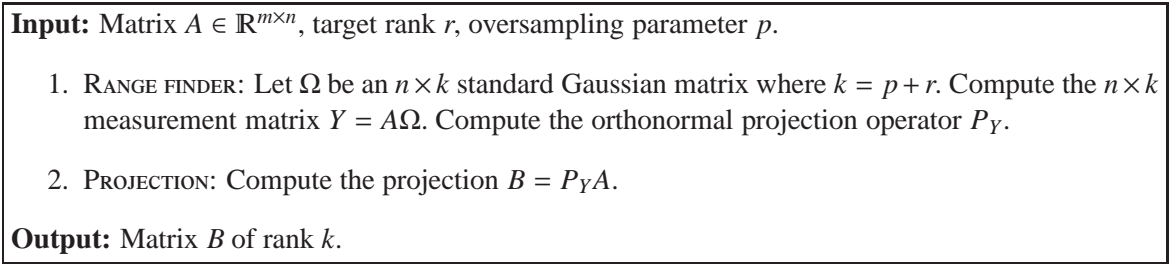

\begin{boxedminipage}{\textwidth}
{\bf Input:} Matrix $A\in\mathbb{R}^{m\times n},$ target rank $r,$ oversampling
parameter $p.$
\begin{enumerate}
\item {\sc Range finder:}
Let $\Omega$ be an $n\times k$ standard Gaussian matrix where $k=p+r.$
Compute the $n\times k$ measurement matrix $Y=A\Omega.$ Compute the
orthonormal projection operator $P_Y.$
\item {\sc Projection:} Compute the projection $B=P_YA.$
\end{enumerate}
{\bf Output:} Matrix $B$ of rank $k.$
\end{boxedminipage}
\caption{Base algorithm for computing a low-rank approximation}
\figurelabel{HMT}
\end{figure}

\begin{theorem}[\cite{HalkoMT11}]
\label{thm:HMT}
Suppose that $A$ is a real $m\times n$ matrix with singular values
$\sigma_1\ge\sigma_2\ge\dots\mper$ Choose a target rank $r\ge 2$ and an
oversampling parameter $p\ge 2$ where $r+p\le\min\{m,n\}.$ Draw an
$n\times(r+p)$ standard Gaussian matrix $\Omega,$ construct the sample
matrix $Y=A\Omega$ and let $B=P_YA.$
Then the expected approximation error in Frobenius norm satisfies
\begin{equation}\label{eq:frob-error}
\E\| A - B \|_F
\le \left(1+\frac r{p-1}\right)^{\nfrac12}\sqrt{\sum_{j>r}\sigma_j^2}\mper
\end{equation}
%Furthermore,
%\begin{equation}\label{eq:spectral-error}
%\E\| A - B \|
%\le \left( 1 + \sqrt{\frac r{p-1}}\right)\sigma_{r+1}
%+ \frac{e\sqrt{r+p}}p\left(\sum_{j>r}\sigma_j^2\right)^{\nfrac12}\mper
%\end{equation}
In particular, for $p=r+1$ we have
\begin{equation}\label{eq:simplified}
\E\| A - B \|_F \le \sqrt{2}\cdot\|A-A_r\|_F
\end{equation}
\end{theorem}
When applying the the theorem we will use Markov's inequality to argue that
the error bounds hold with sufficiently high probability up to a constant
factor loss. As shown in~\cite{HalkoMT11}, much better bounds on the failure
probability are possible. We will omit the precise bounds here for the sake of
simplicity.

\section{Privacy-preserving sub-routines: Range finder and projection}
\sectionlabel{range-projection}

In order to give a privacy preserving variant of the above algorithm, we will
first need to carefully bound the sensitivity of the range finder and of the
projection step, and bound the effect of the necessary perturbations. We do
this for each step in this section.

\subsection{Privacy-preserving range finder}
\sectionlabel{range-finder}

In this section we present a privacy-preserving algorithm which finds a set of
vectors $Y$ whose span contains most of the spectrum of a given
matrix~$A.$

\begin{figure}[ht]
\begin{boxedminipage}{\textwidth}
{\bf Input:} Matrix $A\in\mathbb{R}^{m\times n},$ target rank $r,$ oversampling
parameter $p$ such that $r+p\le\min\{m,n\},$ privacy parameters $\epsilon,\delta\in(0,1).$
\begin{enumerate}
\item Let $\Omega$ be an $n\times k$ standard Gaussian matrix where $k=p+r.$
\item Compute the $n\times k$ measurement matrix $Y=A\Omega.$
\item Let $N\sim N(0,\rho^2)^{m\times k}$ where
$\rho=2\epsilon^{-1}\sqrt{2k\log\left(4k/\delta\right)}.$
%$\rho=2\epsilon^{-1}\sqrt{2k\log\left(\sqrt{\frac{\pi}{2}}\frac{2k}{\delta}\right)\log(2.5/\delta)}.$
\item Let $\tilde Y=Y+N.$
\item Orthonormalize the columns of $\tilde Y$ and let the result be $W.$
\end{enumerate}
{\bf Output:} Orthonormal $m\times k$ matrix $W.$
\end{boxedminipage}
\caption{Privacy-preserving range finder}
\figurelabel{range-finder}
\end{figure}

\begin{lemma}
\label{lem:range-privacy}
The algorithm in~\figureref{range-finder} satisfies $(\epsilon,\delta)$-differential
privacy.
\end{lemma}

\begin{proof}
We argue that outputting $\tilde{Y}$  preserves $(\epsilon,\delta)$-differential privacy. That outputting $W$ preserves $(\epsilon,\delta)$-differential privacy follows from the fact that differential privacy holds under arbitrary post-processing.

Consider any two neighboring matrices $A, A' \in \mathbb{R}^{m\times n}$
differing in their $i$'th row, and let $Y = A\Omega$ and $Y' = A'\Omega$.
Define $e \in \mathbb{R}^n$ to be $e^T = A_{(i)} - A'_{(i)}$. Note that by the
definition of neighboring, we must have $\|e\|_2 \leq 1$. Observe that for
each row $j \neq i$, we have $Y_{(j)} = Y'_{(j)}$, and define $\hat{e} \in
\mathbb{R}^k$ to be $\hat{e} = Y_{(i)} - Y'_{(i)} = e^T\Omega$. First, we will
give a high-probability bound on $\|\hat{e}\|_2$. Observe that for each $j \in
[k],$ $\hat{e}_j$ is distributed like a standard Gaussian:
\[
\hat{e}_j \sim \sum_{\ell = 1}^n e_\ell \cdot N(0,1) = N\left(0, \sum_{\ell =
1}^n e_\ell^2\right) = N(0,1)\mcom
\]
where we used \factref{gaussian-sum}.
Therefore, we have for any $t \geq 1,$
by standard Gaussian tail bounds,
\[
\Pr\Set{\left|\hat{e}_j\right| \geq t} \leq
2\exp\left(-\frac{t^2}{2}\right)
\]
Taking a union bound
over all $k$ coordinates we have:
\[
\Pr\Set{\max_{j\in[k]}\left|\hat{e}_j\right| \geq
\sqrt{2\log\left(4k/\delta\right)}} \leq
\frac{\delta}{2}
\]
In particular, we have except with probability $\delta/2,$
\begin{equation}\equationlabel{hate}
\|\hat{e}\|_2 \leq \sqrt{2k\log\left(4k/\delta\right)}
\end{equation}
Note that we have set $\rho$ such that
conditioned on \equationref{hate} (which holds with probability at
least $1-\delta/2$) we have the following by \theoremref{Gaussian}:
For every set $S\subseteq\mathbb{R}^{m\times k},$
$\Pr\Set{\tilde{Y} \in S} \leq \exp(\epsilon)\Pr\Set{\tilde{Y'} \in S} + \delta/2$.
Hence, without any conditioning we can say:
\[
\Pr\Set{\tilde{Y} \in S} \leq \exp(\epsilon)\Pr\Set{\tilde{Y'} \in S} + \delta
\]
which completes the proof of privacy.
\end{proof}

\begin{theorem}
Let $A$ be an $m\times n$ matrix with singular values
$\sigma_1\ge\sigma_2\ge\dots\mper$
Then, given $A$ and valid parameters $r,p,\epsilon,\delta,$
the algorithm in \figureref{range-finder} returns a matrix $W$ such that
$W$ satisfies $(\epsilon,\delta)$-differential privacy,
and moreover we have the error bound
\begin{equation}
\E\|A-WW^TA\|_F
\le \left(1+\frac r{p-1}\right)^{\nfrac12}\sqrt{\sum_{j>r}(\sigma_j+\rho)^2}
\mper
\end{equation}
\end{theorem}

\begin{proof}
Privacy follows from \lemmaref{range-privacy}. Let us therefore argue the
second part of the theorem. Consider the $m\times (n+m)$ matrix
\[
A' = [ A \mid \rho I_{m\times m} ]\mcom
\]
where $I_{m\times m}$ is the $m\times m$ identity matrix. Let $\Omega'$ denote
a random $(m+n)\times k$ Gaussian matrix.
Note that
\[
\tilde Y \sim A'\Omega'\mper
\]
That is, $\tilde Y$ is distributed the same way as $Y'=A'\Omega'.$ Here, we're
using the fact that $\rho N(0,1)= N(0,\rho^2).$

On the other hand, by \theoremref{HMT}, we know that $Y'$ is a good range for
$A'$ in the sense that
\begin{equation}\label{eq:range-error}
\E\|A'-P_{Y'}A'\|_F
\le \left(1+\frac r{p-1}\right)^{\nfrac12}\sqrt{\sum_{j>r}\sigma_j'^2}\mper
\end{equation}
Here, $\sigma_j'$ denotes the $j$-th largest singular value of $A'.$

\begin{claim}\claimlabel{one}
$ \|A-P_{Y'}A\|_F
\le \|A'-P_{Y'}A'\|_F
$
\end{claim}
\begin{proof}
The claim is immediate, because we can obtain $A$ from $A'$ by truncating the
last $m$ columns. Hence, the approximation error can only decrease.
\end{proof}

\begin{claim}\claimlabel{two}
For all $j,$ we have $ |\sigma_j-\sigma_j'|\le\rho\mper$
\end{claim}

\begin{proof}
Consider the matrix $A_0 = [ A \mid 0_{m\times m} ]$ where $0_{m\times m}$ is
the all zeros matrix. Note that
\[
A' = A_0 + E \quad\text{with}\quad E=[ 0_{m\times n} \mid \rho I_{m\times m}]\mper
\]
Also, $\sigma_j=\sigma_j(A_0),$ since we just appended an all zeros matrix.
On the other hand,
\[
\|E\|_2=\|\rho I_{m\times m}\|_2= \rho\mper
\]
Hence, by Weyl's perturbation bound (\theoremref{weyl})
\[
|\sigma_j-\sigma_j'|
= |\sigma_j(A_0)-\sigma_j'|
\le \|E\|_2=\rho\mper
\]
\end{proof}

Combining the previous claims with (\ref{eq:range-error}), we have
\[
\E\|A-P_{Y'}A\|_F
 \le \E\|A'-P_{Y'}A'\|_F
 \le
\left(1+\frac r{p-1}\right)^{\nfrac12}\sqrt{\sum_{j>r}(\sigma_j+\rho)^2}
\mper
\]
Since $Y'$ and $\tilde Y$ are identically distributed, the same claim is true
when replacing $P_{Y'}$ by $P_{\tilde Y}.$ Furthermore, $P_{\tilde Y}=WW^T$
and so the claim follows.
\end{proof}

\begin{corollary}\corollarylabel{range-finder}
Let $A\in\R^{m\times n}$ be as in the previous theorem. Assume that $m\le n$
and run the algorithm with $p\ge r+1.$ Then, with probability $99/100$,
\[
\|A-WW^TA\|_F
\le O\left(\sqrt{\sum_{j>r}\sigma_j^2} + \sqrt{\rho^2m}\right)\mper
\]
\end{corollary}
\begin{proof}
By Markov's inequality and the previous theorem, we have with probability
$99/100,$
\[
\|A-P_WA\|_F
\le O\left(\sqrt{\sum_{j>r}(\sigma_j+\rho)^2}\right)\mper
\]
But note that $(\sigma_j+\rho)^2=\sigma_j^2+2\sigma_j\rho+\rho^2\le
3\sigma_j^2+3\rho^2.$ This is because either $\sigma_j>\rho\ge1$ and thus
$\sigma_j^2\ge\sigma_j\rho$ or else $\rho\ge\sigma$ in which case
$\rho^2\ge\sigma_j\rho.$ The claim follows by using that
$\sqrt{a+b}\le\sqrt{a}+\sqrt{b}$ for non-negative $a,b\ge0.$
\end{proof}

\subsection{Privacy-preserving projections}
\sectionlabel{projection}

In the previous section we showed a privacy-preserving algorithm that finds a small number of orthonormal vectors $W$ such that $A$ is well-approximated by
$WW^TA.$ To obtain a privacy preserving low-rank approximation algorithm we
still need to show how to carry out the projection step in a
privacy-preserving fashion. We analyze the error of the projection
step in terms of the magnitude of the maximum entry of each column of $W.$ This serves to bound the sensitivity of the matrix multiplication operation. The smaller the entries of $W,$ the smaller the over all error that we incur.

\begin{figure}[h]
\begin{boxedminipage}{\textwidth}
{\bf Input:} Matrix $A\in\mathbb{R}^{m\times n},$ matrix
$W\in\R^{m\times k}$ whose columns have norm at most~$1,$
privacy parameters $\epsilon,\delta\in(0,1).$
\begin{enumerate}
\item Let $W=[w_1\mid w_2\mid\cdots\mid w_k]$ and for each $i\in[k]$
let $\alpha_i=\|w_i\|_\infty$ denote the maximum magnitude entry in $w_i.$
\item Let $N$ be a random $k\times n$ matrix where $N_{ij}\sim N(0,\alpha_i^2\rho^2)$
for $i\in[k],j\in[n]$ and $\rho=2\epsilon^{-1}\sqrt{8k\ln(4k/\delta)\ln(2/\delta)}.$
\item Compute the matrix $B=W(W^TA+N).$
\end{enumerate}
{\bf Output:} Matrix $B$ of rank $k.$
\end{boxedminipage}
\caption{Privacy-preserving projection}
\figurelabel{projection}
\end{figure}

\begin{lemma}
The output $B$ of the algorithm satisfies $(\epsilon,\delta)$-differential
privacy.
\end{lemma}

\begin{proof}
We will argue that releasing $W^TA + N$ preserves
$(\epsilon,\delta)$-differential privacy. That releasing $B$ preserves
differential privacy follows from the fact that differential privacy does not
degrade under arbitrary post-processing.  Fix any two neighboring matrices $A,
A'$ differing in their $i$'th row. Let $E = A - A'$ and
let $e^T = A_{(i)} - A'_{(i)} = E_{(i)}.$
Recall by the definition of neighboring, $\|e\|_2 \leq 1$, and for all
other $j \neq i$, $\|E_j\|_2 = 0$. For any $j \in [k]$, consider the $j$'th
row of $W^TE$:
\[
\|(W^TE)_{(j)}\|_2 = \sqrt{\sum_{\ell=1}^n W_{\ell, j}^2\cdot
e_\ell^2} \leq \alpha_j \|e\|_2 = \alpha_j\mper
\]
Hence, by \theoremref{Gaussian}, releasing $(W^TE)_{(j)} + g^T$ where $g\sim
N(0, \alpha_j^2\rho^2)^n$ preserves $(\frac{\epsilon}{\sqrt{8k\ln
(2/\delta)}}, \frac{\delta}{2k})$-differential privacy.  Finally, we apply
\theoremref{composition} to see that releasing each of the $k$ rows of $W^TE$
preserves $(\epsilon',k(\delta/2k) + \delta/2) =
(\epsilon',\delta)$-differential privacy for: $$\epsilon' \leq
\sqrt{2k\ln(2/\delta)}\cdot\frac{\epsilon}{\sqrt{8k\ln (2/\delta)}} +
2k\left(\frac{\epsilon}{\sqrt{8k\ln (2/\delta)}}\right)^2 \leq \epsilon$$
as desired.
\end{proof}

\begin{theorem}\theoremlabel{projection}
The algorithm above returns a matrix $B$ such that $B$ satisfies
$(\epsilon,\delta)$-differential privacy and moreover with probability
$99/100,$
\[
\|A-B\|_F\le \|A-WW^TA\|_F + O\left(\sqrt{k\sum_{i=1}^k \alpha_i^2\rho^2
n}\right)\mper
\]
In particular if $\max_i\alpha_i=\alpha,$ we have with the same probability,
\[
\|A-B\|_F\le \|A-WW^TA\|_F
+ O\left(\frac{\alpha k\log(k/\delta)\sqrt{n}}\epsilon\right)\mper
\]
\end{theorem}

\begin{proof}
\[
\|A-B\|_F
=\|A-W(W^TA+N)\|_F
=\|A-WW^TA - WN\|_F
\le\|A-WW^TA\|+\|WN\|_F
\]
But $\|W\|_F=\sqrt{k}$ so that, by \factref{sub},
\[
\|WN\|_F\le\|W\|_F\|N\|_F=\sqrt{k}\cdot \|N\|_F.
\]
\Mnote{the previous step doesn't seem tight. I'd like to get rid of the extra
$\sqrt{k}$ so that our algorithm is never worse than randomized response.}
On the other hand, by Jensen's inequality and linearity of expectation,
\[
\E\|N\|_F
\le \sqrt{\E\|N\|_F^2}
= \sqrt{\sum_{i,j} \E N_{ij}^2}
= \sqrt{\sum_{i=1}^k c_k^2\rho^2n}\mper
\]
The claim now follows from Markov's inequality.
\end{proof}

Note that the quantities $\alpha_i$ are always bounded by~$1,$ since all $w_i$'s are unit
vectors. In the next section, we will show that under certain incoherence assumptions, we will have (or will be able to enforce) the condition that the $\alpha_i$ values are bounded significantly below $1$.

%
% insert here
%

\section{Incoherent matrices}
\sectionlabel{incoherent}

Intuitively speaking, a matrix is \emph{incoherent} if its left singular vectors
have low correlation with the standard basis vectors. There are multiple ways to
formalize this intuition. Here, we will work with two natural notions of
coherence. In both cases we will be able to show that we can find---in a
privacy-preserving way---projection operators that have small entries. As
demonstrated in the previous section, this directly leads to improvements over
randomized response.

\subsection{$C$-coherent matrices}

In this section we work with matrices~$A$ in which row norms do not deviate by
too much from the typical row norm.
Another way to look at this condition is that coordinate projections
provide little spectral information about the matrix $A.$ From this angle the
condition we need can be interpreted as \emph{low coherence} in the sense that
the singular vectors of $A$ that correspond to large singular values
must be far from the standard basis.

\begin{definition}[$C$-coherence]
\definitionlabel{C-coherence}
We say that a matrix $A\in\R^{m\times n}$ is \emph{$C$-coherent} if
\[
\max_{i\in[m]}\|e_i^TA\|\le C\cdot\frac{\|A\|_F}{\sqrt{m}}\mper
\]
Note that we have $1\le C\le \sqrt{m}.$
\end{definition}

The next lemma shows that sparse vectors have poor correlation with the matrix
in the above sense. We say a vector $w$ is \emph{$\ell$-sparse} if it has at
most $\ell$ nonzero coordinates.

\begin{lemma}
\lemmalabel{sparse-vector}
Let $A\in\R^{m\times n}$ be a $C$-coherent matrix.
Let $w$ be an $\ell$-sparse unit vector in $\R^m.$
Then,
\[
\|w^TA\| \le \frac{C\sqrt{\ell}\|A\|_F}{\sqrt{m}}\mper
\]
\end{lemma}

\begin{proof}
Since $w$ is $\ell$-sparse we can write it as $w=\sum_{i=1}^\ell\alpha_ie_i$
where $e_1,\dots,e_\ell$ are $\ell$ distinct standard basis vectors
and $\sum_i \alpha_i^2=1.$
Hence,
\[
\|w^TA\| = \left\|\sum_{i=1}^\ell\alpha_i e_i^TA\right\|
\le \sum_{i=1}^\ell \alpha_i\|e_i^TA\|
\le \frac{C\sqrt{\ell}\|A\|_F}{\sqrt{m}}\mper
\]
In the last step we used the Cauchy-Schwarz inequality and the fact that $A$ is
$C$-coherent.
\end{proof}

\begin{lemma}
Let $\alpha>0.$
Let $A\in\R^{m\times n}$ be a $C$-coherent matrix.
Let $w\in\R^m$ be a unit vector and suppose $w_\alpha$ is the vector obtained
from $w$ by zeroing all coordinates greater than $\alpha.$ Then,
\[
w_\alpha^TA = w^TA + e
\]
where $e$ is a vector of norm
\[
\|e\|\le\frac{C\|A\|_F}{\alpha\sqrt{m}}\mper
\]
\end{lemma}

\begin{proof}
Note that $w-w_\alpha$ is an $\ell$-sparse vector with $\ell\le1/\alpha^2.$
Here we used that $w$ is a unit vector and hence there can be at most
$1/\alpha^2$ coordinates larger than $\alpha.$ The lemma now follows directly
from \lemmaref{sparse-vector}.
\end{proof}

The next lemma is a straightforward extension of the previous one for the case
where we multiply $A$ by a matrix $W$ rather than a single vector.
\begin{lemma}\lemmalabel{truncation-error}
Let $\alpha>0.$
Let $A\in\R^{m\times n}$ be a $C$-coherent matrix.
Let $W\in\R^{m\times k}$ be a matrix whose columns have unit length.
Suppose $W_\alpha$ is the matrix obtained
from $W$ by zeroing all entries greater than $\alpha.$ Then,
\[
WW_\alpha^TA = WW^TA + E
\]
where $E$ is a matrix of Frobenius norm
\[
\|E\|_F\le\frac{Ck\|A\|_F}{\alpha\sqrt{m}}\mper
\]
\end{lemma}
\begin{proof}
By the previous lemma,
we have
$W_\alpha^TA = W^TA + E'\mcom$
where every row of $E'$ has Euclidean norm
$C\|A\|_F/\alpha\sqrt{m}.$  Hence,
$\|E'\|_F\le C\sqrt{k}\|A\|_F/\alpha\sqrt{m}.$
But then
\[
WW_\alpha^TA = WW^TA + WE'\mper
\]
Put $E=WE'$ and note that $\|E\|_F\le\|W\|_F\|E'\|_F=\sqrt{k}\|E'\|_F.$ The
lemma follows.
\end{proof}

The previous lemma quantifies what happens if we replace $WW^TA$ by
$WW_\alpha^TA.$ Working with $W_\alpha^T$ for small $\alpha$ instead of $W^T$
will decrease the sensitivity of the computation of $W_\alpha^TA.$  On the
other hand, by the previous lemma we have an expression for the error
resulting from the truncation step.

\subsection{Strong coherence}

Here we introduce and work with the notion of \emph{$\mu_0$-coherence} which
is a standard notion of coherence. As we will see in \sectionref{relation}, it
is a stronger notion than $C$-coherence. Consequently, the results we will be
able to obtain using $\mu_0$-coherence are stronger than our previous results
on $C$-coherence in certain aspects.

\begin{definition}[$\mu_0$-coherence]
\definitionlabel{mu0}
Let $U$ be an $m\times r$ matrix with orthonormal columns and $r\le n.$
Recall, that $P_U=UU^T.$ The \emph{$\mu_0$-coherence} of $U$ is defined as
\begin{equation}
\mu_0(U)
= \frac mr\max_{1\le j\le m} \|P_Ue_j\|^2
= \frac mr\max_{1\le j\le m} \|U_{(j)}\|^2\mper
\end{equation}
Here, $e_j$ denotes the $j$-th $m$-dimensional standard basis vector and
$U_{(j)}$ denotes the $j$-th row of $U.$

The \emph{$\mu_0$-coherence} of an $m\times n$ matrix $A$ of rank $r$ given in
its singular value decomposition $U\Sigma V^T$ where $U\in\R^{m\times r}$ is
defined as $\mu_0(U).$
\end{definition}

\begin{fact}
$1 \le \mu_0(U)\le m$
\end{fact}
\begin{proof}
Since $U$ is orthonormal, there must always exists a row of square norm $r/m.$
On the other hand, no row of $U$ has squared norm larger than $r.$
\end{proof}

The above notion is used extensively throughout the literature in the context
of matrix completion and low rank approximation, e.g., in Candes and Recht
\cite{CandesR09}, Keshavan et al. \cite{KeshavanMO10}, Talwalkar and
Rostamizadeh \cite{TalwalkarR10}, Mohri and Talwalkar \cite{MohriT10}.
Motivated by the Netflix problem, Candes and Tao \cite{CandesT10} study matrix
completion for matrices satisfying a stronger incoherence assumption than
small $\mu_0$-coherence.

Our goal from here on is to show that if we run our range finding algorithm from
\sectionref{range-finder} on a low-coherence matrix it will produce a
projection matrix with small entries. This result (presented in
\lemmaref{mu0-projection}) requires several technical lemmas.

The first technical step is a lemma showing that vectors that lie in the range
of an incoherent matrix must have small $\ell_\infty$-norm.

\begin{lemma}\lemmalabel{mu0-infty}
Let $U$ be an orthonormal $m\times r$ matrix. Suppose $w\in\range(U)$ and
$\|w\|=1.$ Then,
\[
\|w\|_\infty^2\le \frac rm\cdot\mu_0(U)\mper
\]
\end{lemma}

\begin{proof}
Let $u_1,\dots,u_r$ denote the columns of $U.$ By our set of assumptions, there
exist $\alpha_1,\dots,\alpha_r\in\R$ such that
\[
w = \sum_{i=1}^r \alpha_iu_i
\qquad\text{and}\qquad
\sum_i\alpha_{i}^2=1\mper
\]
Therefore, denoting the
$j$-th entry of $w$ by $w_j$ and the $j$-th entry of $u_i$ by $u_{ij},$ we have
\begin{align*}
|w_j|^2
 = \left(\sum_{i=1}^r\alpha_i u_{i j}\right)^2
& \le \left(\sum_{i=1}^r\alpha_i^2\right)
\left(\sum_{i=1}^ru_{i j}^2\right) \tag{by Cauchy-Schwarz} \\
& = \|U_{(j)}\|^2\mper
\end{align*}
In particular $\|w\|_\infty^2 \le \max_{j\in[m]}\|U_{(j)}\|^2.$
On the other hand,
using \definitionref{mu0},
\[
\|w\|_\infty^2\le \max_{j\in[m]}\|U_{(j)}\|^2
= \frac rm\cdot\mu_0(U)\mper
\]
The lemma follows.
\end{proof}

We will need the following geometric lemma: If we start with a small
orthonormal set of vectors of low coherence and we append few random unit
vectors, then the span of the resulting set of vectors has a low coherence
basis.

\begin{lemma}\lemmalabel{mu0-perturb}
Let $u_1,\dots,u_r\in\R^m$ be orthonormal vectors. Pick unit vectors
$n_1,\dots,n_k\in\mathbb{S}^{m-1}$ uniformly at random. Assume that
\begin{equation}\equationlabel{mlarge}
m\ge c_0 k(r+k)\log(r+k)
\end{equation}
where $c_0$ is a sufficiently large constant.
Then, there exists a
set of orthonormal vectors $v_1,\dots,v_{r+k}\in\R^m$ such that
$\mathrm{span}\{v_1,\dots,v_{r+k}\}=\mathrm{span}\{u_1,\dots,u_r,n_1,\dots,n_k\}$ and
furthermore, with probability $99/100,$
\[
\mu_0([v_1\mid \dots\mid v_{r+k}])
\le
2\mu_0([u_1\mid\dots\mid u_k])+ O\left( \frac{k\log m}{r}\right)
\]
\end{lemma}

\begin{proof}
We will construct the basis iteratively using the Gram-Schmidt orthonormalization
algorithm starting with the partial orthonormal basis $u_1,\dots,u_r.$ The
algorithm works as follows: At
iteration $i$ we have obtained a partial orthonormal basis $v_1,\dots,v_t$
where $t=r+i-1.$ We then pick a random unit vector $v\in\mathbb{S}^{m-1}$ and
let $v' = \sum_{i=1}^t v_iv_i^Tv.$ Put
\[
v_{t+1} = \frac{v-v'}{\|v-v'\|}\mper
\]
Let $V_t=[v_1\mid\dots\mid v_t]$ and $V_{t+1}=[V_t\mid v_{t+1}].$
Our goal is to bound $\|v_{t+1}\|_\infty^2$ as this will directly lead to a bound
on $\mu_0(V_{t+1})$ in terms of $V_t.$ Summing up this bound over $t$ will
lead to a bound on $\mu_0(V_{r+k})$ which is what the lemma is asking for.

Let us start with a two simple claims that follow from measure concentration
on the sphere. The first one bounds the $\ell_\infty$-norm of a random unit
vector.
\begin{claim}\claimlabel{one}
$\|v\|_\infty^2 \le O\left(\frac{\log m}m\right)$ with probability
$1-1/200k.$
\end{claim}
\begin{proof}
It is not hard to show that for every $i\in[m],$ the coordinate projection
$f_i(v)=v_i$ is a Lipschitz function on the sphere. Moreover, the median of
$f_i$ is $0$ by spherical symmetry. By measure
concentration (\theoremref{levy}), $\Pr\Set{\left|f_i\right|>
\epsilon}\le O(\exp(-\epsilon^2m/2)).$ Setting $\epsilon=O( \sqrt{(\log
m+\log k)/m})=O(\sqrt{\log(m)/m})$
and taking a union bound over all $m$ coordinates completes the proof.
\end{proof}

The second claim we need bounds the Euclidean norm of $v'.$

\begin{claim}\claimlabel{two}
$\|v'\|^2\le O\left(\frac {r+k} m\right)$ with probability $1-1/200k.$
\end{claim}
\begin{proof}
Proceeding as in proof of the previous claim, we note that for each $i\in[t],$
$f_i(v)=\langle v_i,v\rangle$ is a Lipschitz function on the sphere with
median~$0.$ Applying \theoremref{levy} with $\epsilon=O(\sqrt{(\log t+\log
k)/m}),$ it follows that with probability $1-1/200kt,$
\[
f_i(v)^2\le O\left(\frac{\log k+\log t}m\right)\mper
\]
Taking a union bound over all
$i\in[t]$ we have with probability $1-1/200k,$
\[
\|v'\|^2=\sum_{i=1}^t\langle v_i,v\rangle^2
\le O\left(\frac{t(\log k+\log t)}m\right)
= O\left(\frac{(r+k)\log(r+k))}m\right)
\mcom
\]
where we used that $t\le r+k.$
\end{proof}

On the one hand, note
that $v'$ is in the span of $v_1,\dots,v_t$ by definition. Hence,
\lemmaref{mu0-infty} directly implies
that
\begin{equation}
\equationlabel{three}
\|v'\|_\infty^2 \le\frac{t}{m}\cdot\|v'\|^2\cdot\mu_0(V_t)\mper
\end{equation}
Hence, combining \equationref{three} with \claimref{two},
we have with probability $1-1/200k,$
\begin{equation}\equationlabel{v'infty}
\|v'\|_\infty^2
\le O\left(\frac{t(r+k)\log(r+k)}{m^2}\cdot\mu_0(V_t)\right)\mper
\end{equation}
On the other hand, we can bound $\|v_{t+1}\|_\infty^2$ as follows:
\[
\|v_{t+1}\|_\infty^2
= \frac{\|v-v'\|_\infty^2} {\|v-v'\|^2}
\le \frac{\|v\|_\infty^2+2\|v\|_\infty\|v'\|_\infty+\|v'\|_\infty^2} {\|v-v'\|^2}
\le \frac{3(\|v\|_\infty^2+\|v'\|_\infty^2)} {\|v-v'\|^2}
\mper
\]
By \claimref{two} we have that with probability $1-1/200k,$
\[
\|v-v'\|^2=\|v\|^2+\|v'\|^2-2\langle v,v'\rangle
\ge 1 - 2\|v'\|^2
\ge 1 - O\left(\frac{(r+k)\log(r+k)}m\right)\mper
\]
In the first inequality above we used that $\langle v,v'\rangle = \sum_{i=1}^t \langle
v_i,v\rangle^2=\|v'\|^2.$
We then applied \claimref{two} in the second inequality.
By \equationref{mlarge}, $m$ is sufficiently large so that
\begin{equation}\equationlabel{inverse}
\frac1{\|v-v'\|^2} \le O(1)\mper
\end{equation}

Combining \equationref{v'infty} with \equationref{inverse} and applying
\claimref{one}, we conclude that with
with probability at least $1-1/100k,$
\begin{align*}
\|v_{t+1}\|_\infty^2
&
\le
O \left(\frac{\log m}m +
\frac{t(r+k)\log(r+k)}{m^2}\cdot\mu_0(V_t)
\right) %\\
\end{align*}
But when the above bound on $\|v_{t+1}\|_\infty^2$ holds, then we must
have
\begin{equation}\equationlabel{plusone}
\mu_0(V_{t+1})
\le \mu_0(V_t) + \frac m{t+1}\|v_{t+1}\|_\infty^2
\le \left(1+ O\left(\frac{(r+k)\log(r+k)}{m}\right) \right) \mu_0(V_t) +
O \left(\frac{\log m}{t}\right)
\end{equation}
Taking a union bound over all $k$ steps, we find that with probability
$99/100,$ \equationref{plusone} is true at all steps of the Gram-Schmidt
algorithm. Assuming that this event occurs, we have:
\begin{align*}
\mu_0(V_{r+k})
& \le \left(1+ O\left(\frac{(r+k)\log(r+k)}{m}\right) \right)^k \mu_0(V_r) +
O \left(\frac{k\log m}{r}\right)  \\
& \le 2\mu_0(V_r) + O\left(
\frac{k\log m}{r}\right)
\tag{since $m\gg k(r+k)\log(r+k)$ by \equationref{mlarge}}
\end{align*}
This finishes our proof of the lemma since $\mu_0(V_r)=\mu_0([u_1\mid\dots\mid u_r])$ by
definition.
\end{proof}

The choice of failure probability in the previous lemma was rather arbitrary
and stronger bounds can be achieved.
We finally arrive at the main lemma in this section.

\begin{lemma}\lemmalabel{mu0-projection}
Let $A$ be an $m\times n$ matrix of rank $r.$
Let $\Omega\sim N(0,1)^{n\times k}$ with $k\le r$
denote a random standard Gaussian matrix and define $Y=A\Omega.$
Assume that $m\ge c_0kr\log r$ for sufficiently large constant $c_0.$
Further, let $\sigma>0$ and
$N\sim N(0,\sigma^2)^{m\times k}$ denote a random Gaussian matrix with i.i.d.
entries sampled from $N(0,\sigma^2).$ Put $\tilde Y = A\Omega + N$
and let
$w_1,\dots,w_k$ be an orthonormal basis for the range of $\tilde Y.$
Then, with probability $99/100,$
\[
\max_{i\in[k]}\|w_i\|_\infty
\le
\sqrt{\frac {4r}m\cdot\mu_0(A)}
+ O\left(\sqrt{\frac{k\log m}m}\right)\mper
\]
\end{lemma}

\begin{proof}
Let $U$ denote the left singular factor of $A.$ Let $u_1,\dots,u_r$ denote the
columns of $U.$ We have,
\[
\mathrm{span}(\{w_1,\dots,w_k\})=\range(\tilde Y),
\]
since $w_1,\dots,w_k$ is an orthonormal basis for the range of $\tilde Y$ by
construction. On the other hand,
\[
\mathrm{range}(Y)\subseteq\range(A)=\mathrm{span}(\{u_1,\dots,u_r\})\mper
\]
Since $\tilde Y = Y + N$ this implies that
%\[
$\mathrm{range}(\tilde Y) \subseteq\mathrm{span}\{u_1,\dots,u_r,n_1,\dots
n_k\}\mcom$
%\]
where $n_1,\dots,n_k$ are the columns of $N$ normalized such that $\|n_i\|=1.$
By assumption $m$ is large enough so that we can apply \lemmaref{mu0-perturb}.
Thus we obtain orthonormal vectors $v_1,\dots,v_{r+k}$ satisfying
\[
\mathrm{range}(\tilde Y) \subseteq\mathrm{span}\{v_1,\dots,v_{r+k}\}
\]
and the matrix $V$ whose columns are $v_1,\dots,v_{r+k}$ has coherence
\[
\mu_0(V)\le 2\mu_0(U)
+ O\left(\frac{k\log m}{r}\right)
\]
with probability $99/100.$ In particular,
$w_i \in \range(V)$ for all $i\in[k]\mper$ Therefore,
by \lemmaref{mu0-infty}, we have that
\[
\max_{i\in[k]} \|w_i\|_\infty^2
\le \frac{r+k}{m}\cdot\mu_0(V)
\le
\frac{2(r+k)}{m}\cdot\mu_0(U)
+ O\left( \frac{(r+k)k\log m}{rm}\right)\mper
\]
Since $k\le r$ and $\mu_0(A)=\mu_0(U),$ we conclude that
\[
\max_{i\in[k]} \|w_i\|_\infty
\le
\sqrt{\frac{4r}{m}\cdot\mu_0(A)}
+ O\left(\sqrt{\frac{k\log m}m}\right)\mper
\]
The lemma follows.
\end{proof}

\begin{remark}
We remark that the previous lemma is essentially tight. Indeed, under the
given assumption on $A$ there could be a left singular vector of
$\ell_\infty$-norm $\sqrt{r\mu_0(A)/m}.$ The above lemma implies that
we are never more than a $O(\sqrt{\log m})$-factor away from this bound.
\end{remark}

\subsection{Relation between $C$-coherence and $\mu_0$-coherence}
\sectionlabel{relation}

Here we show that the assumption of small $\mu_0$-coherence is strictly
stronger than that of small $C$-coherence assuming the rank of the matrix is
not too large.

\begin{lemma}
Let $A$ be an $m\times n$ matrix of rank $r.$ Then, $A$ is $C$-coherent where
\[
C \le \sqrt{r\mu_0(A)}\mper
\]
\end{lemma}
\begin{proof}
Let the SVD of $A$ be $U\Sigma V^T$ and denote the right singular vectors by
$v_1,\dots,v_r.$ Extend them arbitrarily to an orthonormal basis of $\R^n,$
denoted $v_1,\dots,v_n.$ We then have for every $j\in[m],$
\begin{equation}\equationlabel{ejTA}
\|e_j^TA\|^2
 =
\sum_{i=1}^n\langle e_j^TA,v_i\rangle^2
 =
 \sum_{i=1}^r \left(\sigma_i\langle e_j, u_i\rangle\right)^2
 \le
\left(\sum_{i=1}^r \left|\sigma_i\langle e_j, u_i\rangle\right|\right)^2\mcom
\end{equation}
where we used that the $\ell_2^2$-norm of a vector is bounded by the
$\ell_1^2$-norm. On the other hand,
\begin{equation}
\left(\sum_{i=1}^r \left|\sigma_i\langle e_j, u_i\rangle\right|\right)^2
\le
\left(\sum_{i=1}^r \left|\sigma_i\right|\left|\langle e_j, u_i\rangle\right|\right)^2
\le
\left(\sum_{i=1}^r\sigma_i^2\right)\left(\sum_{i=1}^r \langle
e_j,u_i\rangle^2\right)
= \|A\|_F^2\cdot\|U_{j}\|^2\mcom
\end{equation}
where we used Cauchy-Schwarz in the inequality. It follows that
\[
\max_{j\in[m]}\|e_j^TA\|^2
= \|A\|_F^2 \max_{j\in[m]} \|U_{j}\|^2
= \|A\|_F^2 \frac{r\mu_0(A)}{m}
\mper
\]
Taking square roots on both sides and rearranging, we find
\[
\frac{\sqrt{m}}{\|A\|_F}\cdot \max_{j\in[m]}\|e_j^TA\|
\le \sqrt{r\mu_0(A)}\mper
\]
Note that the left hand side is exactly the smallest $C$ for which $A$ is
$C$-coherent. This proves the lemma.
\end{proof}

Recall that \lemmaref{sparse-vector} showed that the singular vectors
corresponding to large singular values of a $C$-coherent matrix $A$ cannot be
too sparse. In particular, the top singular vectors must have
small $\mu_0$-coherence as a result. However, we cannot rule out that there
are singular vectors corresponding to small singular values that do have
large coordinates.

\section{Privacy-preserving low rank approximations}

In this section we compose the range finder, projection and truncation step to
get a private low rank approximation algorithm suitable for matrices of low
coherence.

\begin{figure}[ht]
\begin{boxedminipage}{\textwidth}
{\bf Input:} Matrix $A\in\mathbb{R}^{m\times n},$ target rank $r\ge 2,$ oversampling
parameter $p\ge 2,$ pruning parameter $\alpha>0,$
privacy parameters $\epsilon,\delta\in(0,1).$
\begin{enumerate}
\item {\sc Range finder:}
\itemlabel{range-finder}
Run the range finder (\figureref{range-finder})
on $A$ with sampling parameter $k=p+r$ and privacy
parameters $(\epsilon/2,\delta/2).$ Let the output be denoted by~$W.$
\item {\sc Pruning:}
\itemlabel{pruning}
Let $W'$ be the matrix obtained from $W$ by zeroing out all entries larger
than $\alpha.$
\item {\sc Projection:}
\itemlabel{projection}
Run the projection algorithm (\figureref{projection}) on input $A,W'$ and
privacy parameters $(\epsilon/2,\delta/2).$
Let $B$ denote the output of the projection algorithm.
\end{enumerate}
{\bf Output:} Matrix $B$ of rank $k=(r+p).$
\end{boxedminipage}
\caption{The private find and project algorithm (\PFP) for computing
privacy-preserving low-rank approximations}
\figurelabel{find-project}
\end{figure}

\begin{lemma}
The \PFP algorithm satisfies $(\epsilon,\delta)$-differential privacy.
\end{lemma}

\begin{proof}
This follows directly from composition and the privacy guarantee achieved by
the subroutines.
\end{proof}

The next theorem details the performance of \PFP on $C$-coherent matrices. In
particular, it shows that in a natural range of parameters it improves
significantly over randomized response (input perturbation).

\begin{theorem}[Approximation for $C$-coherent matrices]
\theoremlabel{C-approx}
There is an $(\epsilon,\delta)$-differentially private algorithm that
given a $C$-coherent matrix $A\in\R^{m\times n}$ and parameters
$r\ge2,p\ge2$ produces a rank $k=r+p$ matrix $B$ such that with
probability $9/10,$
\begin{equation}\equationlabel{C-precise}
\|A-B\|_F
\le
O\left(
\sqrt{1+\frac r{p-1}}\cdot \|A-A_r\|_F
+
\frac{\sqrt{km}\log(k/\delta)}\epsilon
+\sqrt{C\|A\|_F}k\left(\frac nm\right)^{1/4}
\frac{\log(k/\delta)^{1/2}}{\epsilon^{1/2}}
\right)\mper
\end{equation}
In particular, the second error term is $o\left(\sqrt{kn\log(k/\delta)}/\epsilon\right),$ whenever
\begin{equation}\equationlabel{range-m}
m = o(n) \qquad\text{and}\qquad
\frac{Ck\|A\|_F\sqrt{\log(k/\delta)}}{\sqrt{n}} = o\left(\sqrt{m}\right)\mper
\end{equation}
\end{theorem}
We generally think of $C,k$ as small compared to both $m$ and $n.$
\equationref{range-m} states that the algorithm outperforms randomized
response whenever $m$ is not too large compared to $n$ and not too small
compared to the rank $k,$ the Frobenius norm of $A$ divided by $\sqrt{n},$
and the coherence parameter $C.$ These two conditions are
naturally satisfied for a wide range of parameters. For example, when
$\|A\|_F=O(\sqrt{kn})$ (so that randomized response no longer provides
non-trivial error) and $C=O(1)$ (i.e., the matrix is very incoherent), then
the requirement on $m$ is just that
\[
\omega(k^{3})\le m \le o(n)\mper
\]
The proof of
\theoremref{C-approx} is a straightforward combination of our previous
error bounds for range finding, pruning and projection.
\begin{proof}[Proof of \theoremref{C-approx}]
We run \PFP with the given set of parameters $r,p,\epsilon,\delta$
and a suitable choice of the pruning parameter $\alpha>0.$
Before fixing $\alpha,$ we claim that the error of the algorithm satisfies,
with probability $9/10,$
\[
\|A-B\|_F
\le
O\left(
\sqrt{1+\frac r{p-1}}\cdot \|A-A_r\|_F
+ \frac{Ck\|A\|_F}{\alpha\sqrt{m}}
+\left(\sqrt{km}
+ \alpha k\sqrt{n}\right)
\cdot \frac{\log(k/\delta)}\epsilon\right)
\]
Here, the first term follows from \theoremref{HMT} and an application of
Markov's inequality to argue that the bound holds except with sufficiently
small constant probability. The other terms
follow from \theoremref{projection} (error bound of the
projection algorithm), \corollaryref{range-finder} (error bound of the range
finder), and, \lemmaref{truncation-error} (error bound for the pruning step
with parameter $\alpha$). We can now optimize $\alpha$ so as to achieve the
geometric mean between the two terms that it appears in (as $\alpha$ and
$1/\alpha$). Running \PFP with this choice of $\alpha$ directly results
in the error bound stated in
\equationref{C-precise}. \equationref{range-m} is now easily verified by
equating the $O(\cdot)$-term in \equationref{C-precise} with
$o(\sqrt{kn\log(k/\delta)}/\epsilon)$ and rearranging.

Since all sub-routines fail with probability at most $1/100,$ we can take a
union bound to conclude that the algorithm fails to satisfy the error bound
with probability at most $1/10.$
\end{proof}

We will next analyze the performance of \PFP on $\mu_0$-incoherent matrices. In this
case no truncation is necessary, since we argued that the projection matrix
with high probability already has very small entries. The error bound here is
stronger in certain aspects as we will discuss in a moment.

\begin{theorem}[Approximation for $\mu_0$-coherent matrices]
\theoremlabel{mu-approx}
There is an $(\epsilon,\delta)$-differentially private algorithm that
given a rank $R$ matrix $A\in\R^{m\times n}$ and parameters
$r\ge2,p\ge2$ such that $k=r+p\le R$ and $m\ge \omega(Rk \log R)$
produces a rank $k$ matrix $B$ such that with probability $9/10,$
\begin{equation}\equationlabel{mu-precise}
\|A-B\|_F
\le
O\left(\sqrt{\frac r{p-1}}\cdot \|A-A_r\|_F +
\left(\sqrt{km}
+ \sqrt{\frac{kR\mu_0(A)+k^2\log m}{m}}
\sqrt{kn}\right)
\cdot \frac{\log(k/\delta)}\epsilon
\right)\mper
\end{equation}
In particular, the error is $o\left(\sqrt{kn\log(k/\delta)}/\epsilon\right),$ whenever
\begin{equation}\equationlabel{range-m-2}
m = o(n) \qquad\text{and}\qquad
Rk(\mu_0(A)+\log m)\sqrt{\log(k/\delta)} = o(m)\mper
\end{equation}
\end{theorem}

Just as in the previous theorem we get a range for $m$ in which the algorithm
improves over randomized response. Here, we need the coherence of $A$ to be
small compared to~$m.$ We also observe a dependence on the rank of the matrix.
This means the algorithm presents no improvement if the matrix is close to
being full rank.  Recall that $\mu_0(A)$ can be as small as $O(1).$ In
particular, in the natural case where $\mu_0(A), k, R$ all are small compared
to $m,$ e.g., $m^{0.3},$ the requirement in \equationref{range-m-2} reduces to
$m=o(n).$

Note that \theoremref{mu-approx} is quantitatively stronger than
\theoremref{C-approx} in the following regime: When $k,R,C,\mu_0(A)$ are all
small (e.g., $n^{o(1)}$), $m\le\sqrt{n}$ and $\|A\|_F^2\ge n,$
then \theoremref{mu-approx}
improves over randomized response by a factor of roughly $\sqrt{m},$ whereas
\theoremref{C-approx} achieves an $m^{1/4}$-factor improvement.

\begin{proof}[Proof of \theoremref{mu-approx}]
We run \PFP with the given set of parameters $r,p,\epsilon,\delta$
and $\alpha=1.$ Note that this choice of $\alpha$ implies that we never modify
the matrix returned by the range finder.
We claim that the error of the algorithm is with probability $9/10,$
\[
\|A-B\|_F
\le
O\left(\sqrt{1+\frac r{p-1}}\cdot \|A-A_r\|_F
+ \left(\sqrt{km}
+ \sqrt{\frac{kR\mu_0(A)+k^2\log m}{m}}
\sqrt{kn}\right)
\cdot \frac{\log(k/\delta)}\epsilon
\right)\mcom
\]
which is what we stated in the theorem. The first error term follows as before
from \theoremref{HMT} and Markov's inequality so that it holds with
probability $99/100.$
The term of $O(\sqrt{km\log(k/\delta)}/\epsilon)$
follows from \corollaryref{range-finder}. To understand
the remaining terms that by \lemmaref{mu0-projection} we have that the matrix
$W=[w_1\mid\cdots\mid w_k]$ returned by the range finder satisfies with
probability $99/100,$
\[
\alpha= \max_{i\in[k]}\|w_i\|_\infty\le
\sqrt{\frac {4R}m\cdot\mu_0(A)}
+ O\left(\sqrt{\frac{k\log m}m}\right)\mper
\]
In applying \lemmaref{mu0-projection} we needed that $m\ge c_0kR\log R$ for
sufficiently large constant which is satisfied by our assumption.
Hence, \theoremref{projection}
ensures that the error resulting from the projection
operation is at most $O(\alpha k\sqrt{n\log(k/\delta)}/\epsilon).$ Expanding
$\alpha$ in the latter bound gives the stated error term.
\equationref{range-m} is now easily verified by
equating the $O(\cdot)$-term in \equationref{mu-precise} with
$o(\sqrt{kn\log(k/\delta)}/\epsilon)$ and rearranging.

Again, we can take a union bound over the failure probabilities of the
sub-routines to bound the probability that our algorithm fails to satisfy the
stated bound by $1/10.$
\end{proof}

\bibliographystyle{alpha}
\bibliography{moritz}

\appendix
\section{Privacy of the Gaussian Mechanism}

\begin{theorem}[Gaussian Mechanism]
Let $x, y \in \mathbb{R}^d$ be any two vectors such that $||x-y||_2 \leq c$. Let $Y \in \mathbb{R}^d$ be an independent random draw from $N(0,\rho^2)^d$, where $\rho = c\epsilon^{-1}\sqrt{\log 1.25/\delta}$. Then for any $S \subseteq \mathbb{R}^d$:
$$\Pr[x + Y \in S] \leq \exp(\epsilon)\Pr[y + Y \in S] + \delta$$
\end{theorem}
\begin{proof}
For a set $S \subseteq \mathbb{R}^d$, write $S - x$ to denote the set $\{s - x : s\in S\}$ and $SQ$ to denote $\{sQ : s \in S\}$. Write $S_i = \{s_i : s \in S\}$ to denote the projection of the set onto the $i$'th coordinate of its elements.

First we consider the one dimensional case, where $x, y \in \mathbb{R}$ and $||x-y||_2 = |x-y| \leq c$. Without loss of generality, we may take $x = 0$ and $y = c$. Let $T \subseteq S$ be the set $T = \{z \in S : z < \frac{\rho^2\epsilon}{c} - \frac{c}{2}\}$ First, we argue that $\Pr[x + Y \in S\setminus T] = \Pr[Y \in S \setminus T] \leq \delta$. This follows directly from the tail bound:
$$\Pr[Y \geq t] \leq \frac{\rho}{\sqrt{2\pi}}\exp(-t^2/2\rho^2)$$
Observing that:
$$\Pr[Y \in S \setminus T] \leq \Pr[Y \geq \frac{\rho^2\epsilon}{c} - \frac{c}{2}]$$
and plugging in our choice of  $\rho = c\epsilon^{-1}\sqrt{\log 1.25/\delta}$ completes the claim. Next we show that conditioned on the event that $Y \not\in S \setminus T$, we have:  $\Pr[x + Y \in S] \leq \exp(\epsilon)\Pr[y + Y \in S]$. Conditioned on this event we have:
$$\left|\ln\left(\frac{\Pr[Y \in S]}{\Pr[Y \in S-c]}\right)\right| \leq \max_{z \in T} \left|\ln\left(\frac{\Pr[Y = z]}{\Pr[Y = z-c]}\right)\right| =\left|\ln\left(\frac{\exp(-z^2/2\rho^2)}{\exp(-(z+c)^2/2\rho^2)}  \right)\right|$$
where here $\Pr[Y = t]$ denotes the probability density function of $N(0,\rho^2)$ at $t$.
This quantity is bounded by $\epsilon$ whenever:
$$z \leq \frac{\rho^2\epsilon}{c} - \frac{c}{2}$$
i.e. whenever $z \in T$.
Therefore:
$$\Pr[x + Y \in S] \leq \exp(\epsilon)\Pr[y + Y \in S] + \delta$$
which completes the proof in the 1-dimensional case.

For the multi-dimensional case, we will take advantage of the rotational invariance of the Gaussian distribution to rotate any Euclidean length $c$-perturbation into a length $c$ standard basis vector, reducing it to the $1$-dimensional case.

Consider any two vectors $x, y \in \mathbb{R}^d$ such that $||x-y||_2 \leq c$.  Let $Q \in \mathbb{R}^{d\times d}$ be the orthonormal (rotation) matrix such that $(x-y)Q = c'\cdot e_1$ where $e_1 \in \mathbb{R}^d$ is the 1st standard basis vector $e_1 = (1, 0, \ldots, 0)$, and $c' = ||x-y||_2 \leq c$. We will use the fact that for any orthonormal matrix $Q$, and for any $Y \sim N(0,\rho^2)^d$, $YQ \sim N(0,\rho^2)^d$: i.e. spherically symmetric Gaussian distributions are invariant under rotation. We have:
$$\Pr[x + Y \in S] = \Pr[(x+Y)Q \in SQ] = \Pr[xQ + YQ  \in SQ] = \Pr[Y \in SQ-xQ]$$
We want to bound:
$$\left|\ln\left(\frac{\Pr[Y \in SQ-xQ]}{\Pr[Y \in SQ-yQ]}\right)\right|$$
 Now note that we have chosen $Q$ such that $(SQ - xQ)_i = (SQ-yQ)_i$ for all $i > 1$ (Because $(xQ)_i = (yQ)_i$ for all $y > 1$). Therefore, we have:
 $$\left|\ln\left(\frac{\Pr[Y \in SQ-xQ]}{\Pr[Y \in SQ-yQ]}\right)\right|= \left|\ln\left(\frac{\Pr[Y_1 \in (SQ)_1-(xQ)_1]}{\Pr[Y_1 \in (SQ)_1-(yQ)_1]}\right)\right|$$

Note that by rotational invariance, we have: $\Pr[(zQ)_1 \geq t] = \Pr[z_1 \geq t]$ for any vector $z \in \mathbb{R}^d$, and so we are now again in the $1$-dimensional case, in which the theorem is already proven.

\end{proof}

\section{Measure concentration on the sphere}
In \sectionref{incoherent} we used the following classical result regarding
concentration of Lipschitz functions on the sphere. A proof can be found for
example in Matousek's text book~\cite{Matousek02}.
\begin{theorem}[L\'{e}vy's lemma]\theoremlabel{levy}
Let $f\colon\mathbb{S}^{d-1}\to\R$ be a Lipschitz function in the sense that
\[
\left|f(x)-f(y)\right|\le \|x-y\|_2
\]
and define the median of $f$ as
$\mathrm{med}(f)=\sup\Set{t\in\mathbb{R}\colon\Pr\Set{f\le t}\le\frac12}\mper$
Then,
\[
\Pr\Set{ \left|f-\mathrm{med}(f)\right| > \epsilon}\le 4\exp(-\epsilon^2d/2)\mcom
\]
where probability probability and expectation are computed with respect to the
uniform measure on the sphere.
\end{theorem}

\section{The Netflix Data}
\sectionlabel{netflix}
In this section we illustrate why the data set released by Netflix satisfies the
assumptions underlying \theoremref{informal1}. That is, the matrix is
unbalanced, sparse and $C$-coherent (\definitionref{C-coherence}) for very
small $C.$ Indeed, according to information released by Netflix, the data set
has the following properties:
\begin{enumerate}
\item There are $x=100,480,507$ movie ratings,
$m=17,770$ movies and $n=480,189$ users. In particular, the data set is very sparse
in that only a $x/mn\approx0.011$ fraction of the matrix is nonzero. Also note
that $m\ll n.$
\item The most rated movie in the data set is \emph{Miss Congeniality} with $t=227,715$
ratings (followed by \emph{Independence Day} with $216,233$). Hence, the maximum
number of entries in one row is only a $t/x\approx0.0022$ fraction of the
total number of \emph{nonzero} entries. 
Moreover, all entries of the matrix are in $\{1,\dots,5\}$ and thus very small
numbers.
\end{enumerate}
We conclude that, indeed, the Netflix matrix is \emph{sparse} and the maximum
norm of any row takes up only a tiny fraction of the total norm of the matrix.
We further believe that these properties are likely to hold in other
recommender systems.  Indeed, the average number of ratings per user should be
small (thus resulting in a sparse matrix), and no item should be rated almost
as often as all other items taken together (thus resulting in low coherence).
\end{document}